\documentclass[11pt]{amsart} 
\usepackage{amscd,amssymb,amsxtra}
\usepackage[mathscr]{eucal}% \usepackage{mathrsfs}  % \mathscr{A}
\usepackage{booktabs}
\usepackage{comment}
\usepackage{color}
\usepackage{enumerate}
\makeatletter
\let\@secnumfont\bfseries
\def\section{\@startsection{section}{1}%
  \z@{4\linespacing\@plus\linespacing}{\linespacing}%
  {\bfseries\centering}}
\def\introsection{\@startsection{section}{1}%
  \z@{3\linespacing\@plus\linespacing}{\linespacing}%
  {\bfseries\centering}}
\def\subsection{\@startsection{subsection}{2}%
   \z@{1.25\linespacing\@plus.7\linespacing}{.5\linespacing}%
   {\normalfont\bfseries}}
\def\subsectionsinline{\def\subsection{\@startsection{subsection}{2}%
  \z@{1\linespacing\@plus.7\linespacing}{-.5em}%
  {\normalfont\bfseries}}}

\def\noremarks{\def\rem##1{\relax}\def\refrem##1{\relax}}

\makeatother

\theoremstyle{definition}

\newtheorem{example}[equation]{Example}

\newtheorem*{definition*}{Definition}
\newtheorem*{example*}{Example}
\newtheorem*{problem*}{Problem}
\newtheorem*{exercise*}{Exercise}
\newtheorem*{question*}{Question}
\newtheorem*{construction*}{Construction}

\theoremstyle{remark}

\newtheorem{remark}[equation]{Remark}

\newtheorem*{note*}{Note}
\newtheorem*{notation*}{Notation}
\newtheorem*{remark*}{Remark}
\newtheorem*{data*}{Data}

\theoremstyle{plain}
\newtheorem{theorem}[equation]{Theorem}
\newtheorem{corollary}[equation]{Corollary}

\newtheorem{proposition}[equation]{Proposition}

\newtheorem{assumption}[equation]{Assumption}

\newtheorem*{theorem*}{Theorem}
\newtheorem*{corollary*}{Corollary}
\newtheorem*{lemma*}{Lemma}
\newtheorem*{proposition*}{Proposition}
\newtheorem*{conjecture*}{Conjecture}
\newtheorem*{claim*}{Claim}
\newtheorem*{proposal*}{Proposal}
\newtheorem*{conclusion*}{Conclusion}
\newtheorem*{hypothesis*}{Hypothesis}
\newtheorem*{assumption*}{Assumption}

\numberwithin{equation}{section}

\newcommand{\rem}[1]{\marginpar%
  [$\bf\color{blue}\Rightarrow$]{$\bf\color{blue}\Leftarrow$}[{\tiny 
\color{blue}\bf #1}]} 

\definecolor{refkey}{rgb}{0,.6,.4}
\newcommand{\refrem}[1]{\marginpar%
  [$\bf\color{refkey}\Rightarrow$]{$\bf\color{refkey}\Leftarrow$}[{\tiny 
\color{refkey} \bf #1}]} 

\renewcommand{\:}{\colon}
\DeclareMathOperator{\Aut}{Aut}
\newcommand{\CC}{{\mathbb C}}

\DeclareMathOperator{\End}{End}

\DeclareMathOperator{\id}{id}

\newcommand{\PP}{{\mathbb P}}
\DeclareMathOperator{\pt}{pt}

\newcommand{\RR}{{\mathbb R}}
\newcommand{\TT}{\mathbb T}
\DeclareMathOperator{\Spin}{Spin}

\newcommand{\ZZ}{{\mathbb Z}}

\newcommand{\chiup}{\raise.5ex\hbox{$\chi$}}
\newcommand{\cir}{S^1}

\newcommand{\inv}{^{-1}}
\DeclareRobustCommand{\mstrut}{^{\vphantom{1*\prime y\vee M}}}

\newcommand{\temsquare}{\raise3.5pt\hbox{\boxed{ }}}

\newcommand{\zmod}[1]{\ZZ/#1\ZZ}

\newcommand{\zt}{\zmod2}

\usepackage[all,2cell,dvips]{xy}\renewcommand{\cir}{\ensuremath{S^1}}
\usepackage[colorlinks,citecolor=refkey]{hyperref}

\definecolor{refkey}{rgb}{0,.8,.2}\definecolor{labelkey}{rgb}{1,0,0} 

\DeclareSymbolFont{bbold}{U}{bbold}{m}{n}
\DeclareSymbolFontAlphabet{\mathbbold}{bbold}

\noremarks

\DeclareMathOperator{\Arf}{Arf}
\DeclareMathOperator{\Bord}{Bord}
\DeclareMathOperator{\Cliff}{Cliff}

\DeclareMathOperator{\Det}{Det}

\DeclareMathOperator{\Pfaff}{Pfaff}
\DeclareMathOperator{\Vect}{Vect}
\DeclareMathOperator{\pfaff}{pfaff}

\newcommand{\BFp}{\Bord_{\langle n-1,n,n+1  \rangle}(\sF')}
\newcommand{\BnFp}{\Bord_{\langle n-1,n  \rangle}(\sF')}
\newcommand{\BnF}{\Bord_{\langle n-1,n  \rangle}(\sF)}

\newcommand{\Co}{\Cliff^{\CC}_1}
\newcommand{\GL}{\Gamma \!\mstrut _{L}}
\newcommand{\IRZ}{I_{\RZ}}

\newcommand{\RRZ}{R_{\RZ}}
\newcommand{\RZ}{\RR/\ZZ}

\newcommand{\Sbord}{\Spin\!\Bord}
\newcommand{\Vi}{\Vect_{\textnormal{top}}}

\newcommand{\bW}{\overline{W}}

\newcommand{\cV}{[V]\raise1pt\hbox{\,$\widecheck{\phantom{c}}$}}
\newcommand{\cW}{[W]\raise1pt\hbox{\,$\widecheck{\phantom{c}}$}}

\newcommand{\eff}{\textnormal{eff}}

\newcommand{\sC}{\mathcal{C}}
\newcommand{\sF}{\mathscr{F}}
\newcommand{\sH}{\mathscr{H}}
\newcommand{\sL}{\mathcal{L}}
\newcommand{\sM}{\mathcal{M}}
\newcommand{\sP}{\mathscr{P}}

\newcommand{\triv}{\mathbbold{1}}
\newcommand{\unit}{\bold{1}}

\begin{document}

\abovedisplayskip18pt plus4.5pt minus9pt
\belowdisplayskip \abovedisplayskip
\abovedisplayshortskip0pt plus4.5pt
\belowdisplayshortskip10.5pt plus4.5pt minus6pt
\baselineskip=15 truept
\marginparwidth=55pt

\makeatletter
\renewcommand{\tocsection}[3]{%
  \indentlabel{\@ifempty{#2}{\hskip1.5em}{\ignorespaces#1 #2.\;\;}}#3}
\renewcommand{\tocsubsection}[3]{%
  \indentlabel{\@ifempty{#2}{\hskip 2.5em}{\hskip 2.5em\ignorespaces#1%
    #2.\;\;}}#3} 
%  \indentlabel{\hskip 4em#3}}
\makeatother

\renewcommand{\labelenumi}{\textnormal{(\roman{enumi})}}
\setcounter{tocdepth}{2}

%**end of header

% lasteq@ 49
% lastsec@  6
% lastthm@ 18
% lastfig@  1

%$$\boxed{\boxed{\text{PRELIMINARY VERSION}}}$$\par\vskip 2pc % omit in final

 \title[Anomalies and invertible field theories]{Anomalies and Invertible Field Theories} %% replace \today with short title in final version
 \author[D. S. Freed]{Daniel S.~Freed}
 \thanks{The work of D.S.F. is supported by the National Science Foundation
under grant DMS-1207817.  This work was supported in part by the National
Science Foundation under Grant No. PHYS-1066293 and the hospitality of the
Aspen Center for Physics.} 
 \address{Department of Mathematics \\ University of Texas \\ Austin, TX
78712} 
 \email{dafr@math.utexas.edu}
% \dedicatory{}
 \date{June 20, 2014}
 \begin{abstract} 
 We give a modern geometric viewpoint on anomalies in quantum field theory
and illustrate it in a 1-dimensional theory: supersymmetric quantum
mechanics.  This is background for the resolution of worldsheet anomalies in
orientifold superstring theory.
 \end{abstract}
\maketitle

%\pagestyle{myheadings}   % omit in final
%\markboth{PRELIMINARY VERSION (\today)}{PRELIMINARY VERSION (\today)}  % omit

{\small
\tableofcontents}

   \section{Introduction}\label{sec:1}
% lastsubsec@000

The subject of anomalies in quantum field theories is an old one, and it is
well-trodden.  There is a huge physics literature on this topic of anomalies,
for which one entree is~\cite{Be}.  Important work in the early
1980s~\cite{AS1,AgW,AgG,ASZ} tied the study of \emph{local} anomalies to the
Atiyah-Singer \emph{topological} index theorem, and extensions to
\emph{global} anomalies~\cite{W1,W2} were not far behind.  These ideas were
quickly fit in to \emph{geometric} invariants in index theory, such as the
determinant line bundle and the $\eta $-invariant.  Indeed, many
developments in geometric index theory at that time were directly motivated
by the physics.  A geometric picture of anomalies emerged from this
interaction~\cite[\S1]{F1}.
 
One impetus to reconsider the settled canon on anomalies is a rather sticky
enigma: worldsheet anomalies in Type~II superstring orientifolds.  That was
the subject of my lecture at String-Math~2013, and it will be elaborated
elsewhere.  Here we take the opportunity to introduce a modern geometric
viewpoint on anomalies~(\S\ref{sec:2}), to illustrate it in a simpler
theory~(\S\ref{sec:3}), and to introduce some topology which is crucial in
resolving worldsheet orientifold anomalies~(\S\ref{sec:4}).

The modern point of view rests on the observation that \emph{the anomaly
itself is a quantum field theory}.  It should be expected that anomalies,
which are computed as part of a quantum field theory, obey the locality
principles of quantum field theory.  The anomaly is a very special type of
theory: it is \emph{invertible}.  If in addition an invertible theory is
\emph{topological}, then it reduces to a map of spectra in the sense of
stable algebraic topology.  This presents us with the opportunity to employ
more sophisticated topological arguments.  We remark that an anomalous
quantum field theory is a \emph{relative} quantum field theory~\cite{FT},
related to the anomaly.
 
The simpler theory we revisit here is supersymmetric quantum mechanics~(QM)
with a single supersymmetry.  It was used in the 1980s to give a physics
derivation of the Atiyah-Singer index theorem.  This physical system
describes a particle moving in a Riemannian manifold~$X$.  The quantum
operator which represents the single supersymmetry is the basic Dirac
operator on~$X$, whose definition requires a spin structure.  In the physics
a spin structure is required to cancel an anomaly in the quantization of the
fermionic field.  This is technically much simpler if we assume that $X$~is
even-dimensional and oriented, which we do in~\S\ref{sec:3}.
In~\S\ref{sec:5} we analyze the anomaly without that simplifying assumption.
One consequence is that if $X$~is odd-dimensional, it is most natural to
consider the Hilbert space of the theory to be a module over a complex
Clifford algebra with an odd number of generators.  This is well-known in
differential geometry in the Clifford linear Dirac operator
construction~\cite{LM}, and it seems natural for the physics as well.  (See
Remark~\ref{thm:15}.)
 
I warmly thank Jacques Distler, Greg Moore, Mike Hopkins, and Constantin
Teleman for many years of fruitful collaboration and discussions on topics
related to this paper.  I also thank the referee for his/her careful reading
and useful suggestions.

   \section{Anomalies}\label{sec:2}
% lastsubsec@  3

The reader may wish to consult previous expositions of anomalies
in~\cite{F1}, \cite{F2}, and~\cite{FM}.

  \subsection{Fields and field theories: formal view}\label{subsec:2.1}

An $n$-dimensional\footnote{Here $n$~is the spacetime dimension.  For
supersymmetric quantum mechanics we have~$n=1$.} quantum field theory
  \begin{equation}\label{eq:1}
      Z\:\BnF\longrightarrow \Vi
  \end{equation}
is, formally, a functor from a geometric bordism category of $(n-1)$- and
$n$-dimensional manifolds with fields~$\sF$ to the category of complex
topological vector spaces.  Unraveling this definition we find that to a
closed $n$-manifold---that is, a compact manifold without boundary---the
theory assigns a number~$Z(M)\in \CC$, the \emph{partition function}.  To a
closed $(n-1)$-manifold~$N$ is attached a topological vector space~$Z(N)$,
often called the \emph{quantum Hilbert space}.  The Hilbert space inner
product exists if $Z$~is unitary.\footnote{and if we assume a
\emph{symmetric} formal $n$-dimensional tubular neighborhood of~$N$ is
given.}  Compact $n$-dimensional bordisms map under~$Z$ to continuous linear
maps.  For example, if $M$~is a closed $n$-manifold, and $B_1\cup\cdots\cup
B_r\subset M$ a disjoint union of open $n$-balls, then
  \begin{equation}\label{eq:6}
     Z(M\setminus B_1\cup\cdots\cup B_r)\: Z(S_1)\otimes \cdots\otimes
     Z(S_r)\longrightarrow \CC 
  \end{equation}
encodes correlation functions of local operators, where $S_j=\partial B_j$
and all boundaries are incoming in the bordism.  (In a general quantum field
theory we take a limit as the radii of the balls shrink to zero.)
See~\cite{Se1} for a recent exposition of this geometric definition of
quantum field theory, due to Segal.

The fields~`$\sF$' in~\eqref{eq:1} are, from the point of view of the
theory~$Z$, \emph{background} fields; any \emph{fluctuating} fields have
already been integrated out.  Formally, fields are a simplicial sheaf~$\sF$
on the category of $n$-manifolds and local diffeomorphisms.  Fix a closed
$n$-manifold~$M$.  Then the fields~ $\sF(M )$ on~$M $ form an iterated fiber
bundle.  There are topological fields (orientations, spin structures,
framings, etc.) and geometric fields (scalar fields, metrics, connections,
spinor fields, etc.)  The definition of some fields depends on other fields
(e.g., a spinor field depends on a metric and spin structure), which is why
$\sF(M )$~is an iterated fibration and not a Cartesian product.  Some fields
have internal symmetries, and so $\sF(M)$ is typically an infinite
dimensional \emph{higher stack}.  Examples of fields with internal symmetries
include spin structures, connections (gauge fields), and higher gauge fields
such as the $B$-field in string theory.  The sheaf condition encodes the
locality of fields and allows the construction of a bordism category with an
arbitrary collection of fields.  The manifolds~$M,N,B_j,S_j$ in the previous
paragraph and going forward are assumed endowed with fields, though the
fields are not always made explicit in the notation.

A field theory~$\alpha $ is \emph{invertible} if for every closed
$(n-1)$-manifold ~$N$ with fields the vector space~$\alpha (N)$ is a line and
if for each $n$-dimensional bordism $M\:N_0\to N_1$ with fields the linear
map $\alpha (M)\:\alpha (N_0)\to \alpha (N_1)$ is invertible.  In particular,
the partition function $\alpha (M)\in \CC$ of a closed $n$-manifold is
nonzero.  The natural algebraic operation on field theories is
multiplication---tensor product of the quantum vector spaces and numerical
product of the partition functions---and `invertibility' refers to that
operation.  For example, the vector space~$\CC$ is the identity under tensor
product of vector spaces, and a vector space~$V$ has an inverse~$V'$---i.e.,
there exists an isomorphism $V\otimes V'\xrightarrow{\;\cong \;}\CC$---if and
only if~$\dim V=1$.

A \emph{lagrangian theory} is specified by a collection of
fields~$\sF$---both background and fluctuating---and, for each
$n$-manifold~$M$, a function
  \begin{equation}\label{eq:2}
     A=A(M )\:\sF(M )\longrightarrow \CC
  \end{equation}
called the \emph{exponentiated action}.  Note that despite the name, there is
not necessarily a well-defined action which would be its logarithm.   

  \begin{example}[]\label{thm:1}
 In supersymmetric~QM with values in a fixed Riemannian manifold~$X$, the
manifold~ $M$~is 1-dimensional and $\sF(M)$~consists of 4~fields: a metric
on~$M$, a spin structure on~$M$, a smooth map $\phi \:M\to X$, and a spinor
field~$\psi $ on~$M$ with values in~$\phi ^*TX$.  The metric, spin structure,
and $\phi $~are independent of each other, but we need all three to define
the space of spinor fields~$\psi $.  Also, the fermionic field~$\psi $ is odd
in the sense of supermanifolds~\cite{DM}, so the exponentiated
action~\eqref{eq:2} is not really a complex-valued function on fields, but as
we only consider bosonic fields in the sequel we do not dwell on this.
  \end{example}

 If the fields~$\sF$ include fermionic fields, as in supersymmetric~QM, then
there is an odd vector bundle $\sF\to\sF'$ with fibers the fermionic fields
and base the bosonic fields.  The fermionic fields can be integrated out to
give a theory with only bosonic fields~$\sF'$.  Each fermionic path integral
contributes the pfaffian of a Dirac operator to the effective exponentiated
action~$A_{\eff}=A_{\eff}(M )$ on~$\sF'(M )$.  The pfaffian may vanish, so
$A_{\eff}$~is not necessarily an invertible theory.  The Feynman procedure next
calls for integration of~$A_{\eff}$ over the bosonic fields~$\sF'(M )$, and
this brings in all the analytic interest of quantum field theory: one needs
to construct a well-defined measure on~$\sF'(M )$ to define the integral.

  \subsection{Anomalies: traditional view}\label{subsec:2.2}

The anomaly is a geometric, rather than analytic, obstruction to
integrating~$A_{\eff}$ over~$\sF'(M )$.  Namely, it may happen
that rather than a global function, the effective exponentiated
action~$A_{\eff}$ is a section of a complex line bundle
  \begin{equation}\label{eq:3}
     \alpha (M )\longrightarrow \sF'(M ). 
  \end{equation}
Furthermore, in a unitary theory $\alpha (M )$~ carries a hermitian metric
and compatible covariant derivative.  Typically $\alpha (M)$~is a tensor
product of more primitively defined line bundles.  For example, if
$A_{\eff}$~is obtained by integrating out fermionic fields, then some factors
of~$\alpha (M )$ are Pfaffian line bundles of families of Dirac operators
parametrized by~$\sF'(M )$.  To obtain a function to formally integrate
over~$\sF'(M )$ we require a \emph{setting of the quantum integrand}, a
section~$\triv$ of~\eqref{eq:3} which we demand be flat and have unit norm.
Then the desired quantum integrand is the ratio~$A_{\eff}/\triv$.

From this lagrangian point of view, the anomaly is the obstruction to the
existence of~$\triv$.  The \emph{local anomaly} is the curvature
of~\eqref{eq:3}; if the curvature vanishes, the \emph{global anomaly} is the
holonomy.  If all holonomies are trivial, then the local and global anomalies
vanish.  Vanishing holonomy implies the existence of~$\triv$, though $\triv$~
is unique only up to a phase on each component of~$\sF'(M)$.  Said
differently, the set of trivializations on each component is a torsor over
the circle group of unit norm complex numbers.
 
There is also a hamiltonian point of view on anomalies~\cite{Se2}, \cite{Fa},
\cite{NAg}.  To an\footnote{An object of~$\BnFp$ is really the germ of an
$n$-manifold neighborhood of~$N$ and the fields are defined on that
neighborhood.}  $(n-1)$-dimensional manifold~$N$ a \emph{non-anomalous} field
theory assigns a fiber bundle over~$\sF'(N)$ whose fibers are complex
topological \emph{vector} spaces.  In an \emph{anomalous} theory~$F$ the
fibers of the bundle
  \begin{equation}\label{eq:4}
     F(N)\longrightarrow \sF'(N) 
  \end{equation}
are rather complex \emph{projective} spaces.  This is in line with
expectations in quantum mechanics: the space of pure states in a quantum
system is a complex projective space.  ``Integrating'' over bosonic fields,
or canonical quantization, involves taking $L^2$~sections of a vector
bundle.\footnote{In fact, one takes sections over the space of classical
solutions which are flat along some polarization, but here we only focus on
the formal geometric difficulty to do with projectivity of the fibers, not
the polarization.}  Again there is an analytic difficulty---construct a
measure on the space of bosonic fields---and a geometric difficulty---lift
the projective bundle to a vector bundle.  The obstruction
  \begin{equation}\label{eq:18}
     \alpha (N) \longrightarrow \sF'(N)
  \end{equation}
to the existence of a lift is the anomaly.  Topologically, this obstruction
is a twisting of complex $K$-theory, or a gerbe (see~\cite{FHT,ASe}, for
example).  It describes a twisted notion of `complex vector bundle', exactly
as a complex line bundle describes a twisted notion of `complex-valued
function'.  In a unitary theory there is also differential geometry---the
obstruction is a ``differential twisting'' of complex $K$-theory---just as in
a unitary theory the line bundle~\eqref{eq:3} carries a metric and
connection.  For example, the local hamiltonian anomaly is measured by a
3-form on~$\sF'(N)$.

A \emph{hamiltonian setting} is a trivialization of the
anomaly~\eqref{eq:18}.  If the anomaly vanishes, then on each component
of~$\sF'(N)$ the trivializations form a torsor over the Picard groupoid of
flat hermitian line bundles.

  \subsection{Anomalies: modern view}\label{subsec:2.3}

As quantum field theory is local on spacetime, we require that the
bundles~$\alpha (M )$ and~$\alpha (N)$ be \emph{local} functions of~$M $
and~$N$.  The same is required for trivializations of anomalies.  Locality is
encoded by demanding that the anomalies~\eqref{eq:3} and~\eqref{eq:18} fit
together as parts of an \emph{invertible extended} $(n+1)$-dimensional field
theory\footnote{If $\alpha $~is unitary and \emph{not} topological, then we
promote~$\alpha $ to a differential field theory in the sense that the line
bundles and gerbes are smooth over smooth parameter spaces and carry metrics
and connections.  In supersymmetric~QM the anomaly is topological, so we will
not pursue this here and tacitly assume that $\alpha $~is topological.}
  \begin{equation}\label{eq:5}
     \alpha \:\BFp\longrightarrow \Sigma ^{n+1}\IRZ. 
  \end{equation}
`Extended' means that $\alpha $~has values on manifolds with corners of
dimensions~$n+1$, $n$, and~$n-1$.  We remark that the numerical invariants of
closed $(n+1)$-manifolds include the holonomies of the anomaly line
bundle~\eqref{eq:3}.  There is flexibility in choosing the codomain
in~\eqref{eq:5}.  Here we take a universal choice, the Pontrjagin or
Brown-Comenetz dual~$\IRZ$ of the sphere spectrum~\cite[Appendix~B]{HS},
shifted up in degree.  In~\S\ref{subsec:3.4} and~\S\ref{sec:5} we make more
economical choices.  After exponentiation: $\alpha (W)$~is a complex number
of unit norm for a closed $(n+1)$-manifold~$W$; $\alpha (M)$~is a
$\zt$-graded complex line for a closed $n$-manifold~$M$; and $\alpha (N)$~is
a gerbe with various $\zt$-gradings for a closed $(n-1)$-manifold~$N$.  Since
the theory is invertible, \eqref{eq:5}~factors though the quotient of the
bordism 2-category obtained by inverting all morphisms.  As the bordism
category is symmetric monoidal what is obtained is a \emph{spectrum} in the
sense of algebraic topology; see~\cite[\S2.5]{L}.  A theorem of
Galatius-Madsen-Tillmann-Weiss~\cite{GMTW} identifies it as an unstable
approximation to a Thom spectrum.  For the anomaly of supersymmetric quantum
mechanics, there are non-topological fields---the metric and the map to the
target---so it is not automatic that the anomaly is topological.
Nonetheless, it is.  In particular, the factorization of~\eqref{eq:5} is a
map of spectra, so is amenable to analysis via techniques of homotopy theory.

As stated earlier, an anomalous theory~$F$ is an example of a \emph{relative
quantum field theory}~\cite{FT}.  Thus it is a map
  \begin{equation}\label{eq:44}
     F\:\unit\longrightarrow \tau \mstrut _{\le n}\alpha 
  \end{equation}
of $n$-dimensional field theories from the trivial theory to the
$n$-dimensional truncation of~$\alpha $.  To a closed $n$-manifold~$M$ with
fields it attaches an element~$F(M)$ of the complex line~$\alpha (M)$, and to
a closed $(n-1)$-manifold~$N$ with fields it attaches a complex vector space
$F(N)$ twisted by the gerbe~$\alpha (N)$.
 
The anomaly is trivializable if $\alpha $~is isomorphic to the trivial
theory, and a trivialization of the anomaly, or setting, is a choice of
isomorphism 
  \begin{equation}\label{eq:48}
     \triv\:\alpha \xrightarrow{\;\cong \;}\unit 
  \end{equation}
\emph{as field theories}.  This general formulation encodes the locality of the
setting of the quantum integrand as well as the locality of the anomaly
itself.

   \section{Supersymmetric quantum mechanics}\label{sec:3}
% lastsubsec@  4

Supersymmetric quantum mechanics (QM) with minimal supersymmetry was used
in~\cite{Ag, FWi} to give a physics derivation of the Atiyah-Singer index
theorem for a single Dirac operator.  An account geared to mathematicians
appears in~\cite{W3}, and a mathematically precise take on the argument was
given in~\cite{Bi}, inspired by~\cite{At1}.  We restrict our attention here
to the anomaly and its trivialization, which is a prerequisite to having a
well-defined quantum mechanical theory.
 
Supersymmetric QM is a 1-dimensional theory of a particle moving in a
Riemannian manifold~$X$.  The theory is defined on 1-manifolds~$M$ equipped
with a background metric and spin structure.  There are two fluctuating
fields on~$M$ which are integrated out in the quantum theory.  First, a map
$\phi \:M\to X$ which represents the trajectory of a particle.  Then there is
an odd field~$\psi $ which is a spinor field on~$M$ with coefficients in the
pullback tangent bundle $\phi ^*TX\to M$.  The lagrangian
density~\cite[pp.~647--656]{Detal} has kinetic terms for these fields:
  \begin{equation}\label{eq:9}
     \sL = \frac 12 \left\{ \langle \frac{d\phi }{dt},\frac{d\phi}{dt}
     \rangle + \langle \psi ,D\psi \rangle \right\} |dt|, 
  \end{equation}
where $t$~is a local coordinate on~$M$ with $d/dt$~of unit length and $D$~is
the Dirac operator on~$M$, coupled to the pullback bundle $\phi ^*TX\to M$.
A spin structure on~$M$ can be identified as a real line bundle $L\to M$
equipped with an isomorphism $L^{\otimes 2}\xrightarrow{\;\cong \;}T^*M$.
Multiplication and integration over~$M$, assuming $M$~is closed, gives a
self-dual pairing on spinor fields with respect to which the Dirac
operator~$D$ is formally skew-adjoint.  The spinor fields, which are sections
of $L\to M$, are real, as is the skew-adjoint Dirac operator.  We do not
dwell on the precise meaning of the kinetic term for fermions.

  \subsection{Lagrangian anomaly}\label{subsec:3.1}

Integrate out the fermionic field~$\psi $, assuming that the 1-manifold
$M$~is closed.  In the notation of~\S\ref{sec:2} this is fermionic
integration over the fibers of $\sF(M)\to\sF'(M)$.  The result is standard:
ignoring the kinetic term for~$\phi $, which plays no role in anomaly
analysis, we obtain the pfaffian of the Dirac operator, which is a
section~$\pfaff D$ of the Pfaffian line bundle~\cite[\S3]{F3}
  \begin{equation}\label{eq:10}
     \Pfaff D\longrightarrow \sF'(M). 
  \end{equation}
Furthermore, this \emph{real} line bundle carries a metric and compatible
covariant derivative.  Thus locally there are two unit norm sections~$\triv$;
an orientation of~$\Pfaff D\to\sF'(M)$---which is a topological
trivialization and may not exist---picks out a global section.

In this section we make the following hypothesis, which we relax
in~\S\ref{sec:5}.

  \begin{assumption}[]\label{thm:2}
 The target manifold~$X$ is even-dimensional and oriented. 
  \end{assumption}

  \begin{theorem}[]\label{thm:3}
 Given Assumption~\ref{thm:2}, the topological equivalence class in
$H^1(\sF'(M);\zt)$ of the lagrangian anomaly $\Pfaff D\to\sF'(M)$ is the
transgression of~$w_2(X)\in H^2(X;\zt)$.
  \end{theorem}

\noindent
 Because $\sF'(M)$ includes the field $\phi \:M\to X$, there is an evaluation
map which is the top arrow in the diagram 
  \begin{equation}\label{eq:38}
     \begin{gathered} \xymatrix{\sF'(M)\times M\ar[d]_{\pi_1}\ar[r]^<<<<<e &X \\
     \sF'(M)} \end{gathered} 
  \end{equation}
The vertical map is projection onto the first factor.  Transgression is the
composition $(\pi _1)_*\circ e^*$ on mod~2 cohomology.  The pushforward 
  \begin{equation}\label{eq:49}
     (\pi _1)_* \:H^2\bigl(\sF'(M)\times M;\zt \bigr)\longrightarrow
     H^2\bigl(\sF'(M);\zt \bigr) 
  \end{equation}
in mod~2 cohomology is defined without any orientation data on the fibers
of~$\pi _1$.  Notice that the anomaly is purely topological; it is
independent of the background metric on~$M$.  It also turns out to be
independent of the background spin structure on~$M$, as is clear from the
formula in the theorem.  Theorem~\ref{thm:3} is well-known.  The proof we
sketch here, which is based on the topological Atiyah-Singer index theorem,
appears in~\cite[(5.22)]{FW}.

  \begin{proof}
 The manifold~ $M$ is a finite union of circles, and since under disjoint
union $\Pfaff D$~ is multiplicative and the transgression of~$w_2(X)$ is
multiplicative, it suffices to consider~$M=\cir$.  Also, the class
  \begin{equation}\label{eq:11}
     [\Pfaff D]\in H^1\bigl(\sF'(M);\zt \bigr) 
  \end{equation}
is determined by its pairing with the fundamental class of smooth loops
$\cir\to\sF'(M)$.  Pull back~\eqref{eq:38} over a single loop to obtain a
family
  \begin{equation}\label{eq:12}
     \begin{gathered} \xymatrix{\cir\times M \ar[r]^<<<<e \ar[d]_<<<<{\pi
     _1}&X\\\cir }
     \end{gathered}
  \end{equation}
of circles parametrized by the circle.  The Atiyah-Singer theorem~\cite{AS2}
computes the value of~\eqref{eq:11} on the base circle as a pushforward
in~$KO$-theory, where the base circle has the bounding spin structure: 
  \begin{equation}\label{eq:13}
     \bigl\langle [\Pfaff D],[\cir] \bigr\rangle = \pi ^{\cir\times M}_*(e
     ^*TX)\in KO^{-2}(\pt)\cong \zt. 
  \end{equation}
No matter what the spin structure on the circle~$M$, the torus~$\cir\times M$
has the bounding spin structure, whence \eqref{eq:13}~is independent of the
spin structure on~$M$.  For the bounding spin structure the $KO$-pushforward
$\pi _*^{\cir\times M}$ of the trivial bundle vanishes, so we can replace~$e
^*TX$ by the reduced virtual bundle, and now by excision we replace the
torus~$S^1\times M$ with the 2-sphere.  Then the $KO$-pushforward becomes the
suspension isomorphism, and since $\widetilde{KO}^{0}(S^2)\cong \zt$ via the
second Stiefel-Whitney class, it follows that
  \begin{equation}\label{eq:14}
     \bigl\langle [\Pfaff D],[\cir] \bigr\rangle = \bigl\langle \phi
     ^*w_2(X),[S^1\times M] 
     \bigr\rangle, 
  \end{equation}
as desired. 
  \end{proof}

  \begin{remark}[]\label{thm:5}
 The Atiyah-Singer index theorem computes the Pfaffian line bundle as a
transgression in $KO$-theory.  Since $M$~has very small dimension, and because
we make the simplifying Assumption~\ref{thm:2}, a very simple truncation of
$KO$-theory suffices, namely mod~2 cohomology.  When we drop
Assumption~\ref{thm:2} in~\S\ref{sec:5} the Pfaffian will be computed by a
somewhat larger truncation of~$KO$-theory.
  \end{remark}

  \begin{remark}[]\label{thm:4}
 The lagrangian anomaly is a complex line bundle, the complexification
of~\eqref{eq:10}, so its equivalence class in~$H^2\bigl(\sF'(M);\ZZ \bigr)$
is the integral Bockstein of the equivalence class of the real
bundle~\eqref{eq:10}.  Since integral Bockstein~$\beta \mstrut _{\ZZ}$
commutes with transgression, that equivalence class is the transgression of
  \begin{equation}\label{eq:21}
     W_3(X)=\beta \mstrut _{\ZZ}w_2(X)\;\in H^3(X;\ZZ). 
  \end{equation}
But since supersymmetric~QM is unitary, the anomaly bundle carries a metric
and connection.  In this case the connection is flat of order two---all
holonomies are~$\pm1$---and is encoded precisely by the real structure, i.e.,
by the real Pfaffian line bundle~\eqref{eq:10}.
  \end{remark}

  \subsection{Hamiltonian anomaly}\label{subsec:3.2}

For more details on parts of this subsection, see~\cite[pp.~372--373]{Detal}
and \cite[pp.~679--681]{Detal}.

It suffices to consider a connected 0-manifold, so a point~$N=\pt$.
Technically, we should embed~$N$ in a germ of a Riemannian 1-manifold, but
that plays no role since ultimately the anomaly is topological.  We also have
a spin structure on the augmented tangent bundle, augmented in the sense that
we add a trivial bundle to make it 1-dimensional.  Up to isomorphism this is
determined by a sign, comparing the orientation underlying the spin structure
to the standard orientation on the real line~$\RR$.  We take the sign to
be~$+$.  The space of classical solutions to the Euler-Lagrange equations
derived from the lagrangian~\eqref{eq:9} is a symplectic supermanifold, and
for the partial quantization which integrates out the fermionic field~$\psi $
we work with a fixed~$\phi $.  In canonical quantization we only
consider~$\phi,\psi $ which satisfy the classical equations of motion, a
second order ODE for~$\phi $ and a first order ODE for~$\psi $.  The space of
classical solutions~$\phi ,\psi $ on~$\RR\times N$ (time cross space) may,
after choosing an initial time, be identified with the supersymplectic
manifold
  \begin{equation}\label{eq:15}
     \pi ^*\Pi TX\longrightarrow TX, 
  \end{equation}
where $\pi \:TX\to X$ is the tangent bundle with its symplectic structure
derived from the Riemannian metric, via the induced isomorphism
$TX\xrightarrow{\;\cong \;}T^*X$ and the standard symplectic structure on the
cotangent bundle.  The fibers of~\eqref{eq:15} are the parity-reversed
tangent spaces, which have an odd symplectic structure given by the
Riemannian metric.  The quantization problem for the constant $\phi \equiv x$
is to quantize the odd symplectic vector space~$\Pi T_xX$.
Assumption~\ref{thm:2} that $X$~is even dimensional ensures the existence of
a complex polarization, which is the parity reversal of a half-dimensional
isotropic subspace~$W\subset T_xX\otimes\CC$ of the complexified tangent
space.  This induces a complex structure on~$T_xX$, and we demand that the
induced orientation agree with the orientation given in
Assumption~\ref{thm:2}.  Write the polarization as a decomposition
  \begin{equation}\label{eq:16}
     T_xX\otimes \CC\cong W\oplus \bW. 
  \end{equation}
The quantum Hilbert space is then the space of functions on~$\Pi W$, which we
identify with the $\zt$-graded exterior algebra~$\sH={\textstyle\bigwedge}
W^*\cong {\textstyle\bigwedge} \bW$.  Complex linear functions on~$\Pi T_xX$
act as operators on~$\sH$: elements of~$(\Pi W)^*\cong \Pi \bW$ act by
exterior multiplication and elements of $(\Pi \bW)^*\cong \Pi W$ by
contraction.  These are the standard creation and annihilation operators, and
they generate the action of the Clifford algebra built on~$T_xX^*\otimes
\CC$.
 
The Clifford module~$\sH$ depends on the choice of
polarization~\eqref{eq:16}.  The underlying projective space~$\PP\sH$ is
independent of the polarization.  Thus, without any choice of polarization,
partial hamiltonian quantization along the fibers of~\eqref{eq:15} produces a
bundle $\pi ^*\sP\to TX$ of complex projective spaces, where
  \begin{equation}\label{eq:17}
     \sP\longrightarrow X 
  \end{equation}
is the bundle of projective complex spin representations.  In other words, if
$SO(X)\to X$ is the oriented orthonormal frame bundle with structure
group~$SO_{2m}$, then \eqref{eq:17}~is the bundle associated to the
projective spin representation $SO_{2m}\to\Aut(\PP)$.  The projective
bundle~\eqref{eq:17}, pulled back to~$TX$, is one model for the hamiltonian
anomaly~\eqref{eq:18}.  Another model is the pullback of the bundle of
complex Clifford algebras
  \begin{equation}\label{eq:19}
     \Cliff^\CC(TX)\longrightarrow X ,
  \end{equation}
formed as the associated bundle to the conjugation action
$SO_{2m}\to\Aut(\Cliff^\CC_{2m})$ on the standard complex Clifford algebra.
 
The bundles~\eqref{eq:17} and~\eqref{eq:19} are both standard models for the
\emph{gerbe} represented by the integral Bockstein~\eqref{eq:21} of the
second Stiefel-Whitney class of~$X$.  As in Remark~\ref{thm:4} the
hamiltonian gerbe carries flat differential geometric data which amount to
the \emph{real} gerbe represented by the bundle of real Clifford algebras
  \begin{equation}\label{eq:22}
     \Cliff(TX)\longrightarrow X . 
  \end{equation}
Its equivalence class is precisely the second Stiefel-Whitney class
$w_2(X)\in H^2(X;\zt)$.  From the field theory point of view, a fiber
of~\eqref{eq:19} is the operator algebra generated by~$\psi $.

  \subsection{Trivializing the lagrangian and hamiltonian
  anomalies}\label{subsec:3.3} 

We show that a spin structure on~$X$ induces a canonical trivialization of
the lagrangian anomaly~\eqref{eq:10} and the hamiltonian
anomaly~\eqref{eq:22}, where for both we incorporate the real structures. 
 
The statement for the lagrangian anomaly follows from a ``categorification''
of Theorem~\ref{thm:3}.  Namely, Theorem~\ref{thm:3} is a topological formula
for the equivalence class of the lagrangian anomaly; it is a topological
index theorem.  What we construct now is an \emph{isomorphism} of the
Pfaffian line bundle with a real line bundle which represents the
transgression of~$w_2(X)$.  A spin structure on~$X$ induces a trivialization
of this line bundle and so, via this isomorphism, a trivialization of the
lagrangian anomaly.  The argument appears in~\cite[\S5.2]{DFM} for the
nonbounding spin structure on the circle; here we give a few more details and
treat the bounding spin structure as well.

The Pfaffian line bundle~\eqref{eq:10} carries a Quillen metric.  The points
of unit norm in each fiber form a $\zt$-torsor, and from the torsor we can
canonically reconstruct the fiber as a real line with metric.  The torsor is
canonically equivalent to~$\pi _0\bigl(\Pfaff D_\phi \setminus \{0\}\bigr)$,
where $\Pfaff D_\phi $~is the fiber over~$\phi $.  As in~\S\ref{subsec:3.1}
it suffices to take~$M=\cir$.  Fix $\phi \:\cir \to X$ and let $E=\phi
^*TX\to \cir $ be the pullback tangent bundle, which is an oriented real
vector bundle with metric and covariant derivative.  Let $SO(E)\to \cir $ be
its bundle of oriented orthonormal frames.  Let $\Gamma_\phi $~be the space
of sections of $SO(E)\to \cir $, which is nonempty.  Since the group of
homotopy classes of maps $\cir \to SO_{2m}$ is cyclic of order two, $\pi
_0(\Gamma_\phi )$~is a $\zt$-torsor.

  \begin{theorem}[]\label{thm:7}
 After a universal choice, there is a canonical isomorphism 
  \begin{equation}\label{eq:45}
     \pi _0\bigl(\Pfaff D_\phi\setminus \{0\}\bigr) \xrightarrow{\;\cong
     \;}\pi _0(\Gamma_\phi ). 
  \end{equation}
\end{theorem}

\noindent
 The universal choice is a path in the special orthogonal group from~$1$
to~$-1$; cf., Remark~\ref{thm:9}.

  \begin{corollary}[]\label{thm:8}
 A spin structure on~$X$ determines a trivialization~$\triv$ of $\Pfaff
D_\phi $.
  \end{corollary}

\noindent 
 The fact that the trivialization is canonical, given the spin structure
on~$X$, means that the trivializations of the lines~$\Pfaff D_\phi $ patch to
a smooth trivialization of the lagrangian anomaly~\eqref{eq:10}.

  \begin{proof}[Proof of Corollary~\ref{thm:8}]
 A spin structure on~$X$ induces a spin structure on $E\to \cir $.  Let
$\Spin(E)\to SO(E)\to \cir $ be the corresponding $\Spin_{2m}$-bundle of
frames.  The space of sections of $\Spin(E)\to \cir $ is connected, so maps
into a single component of~$\Gamma_\phi $.
  \end{proof}

  \begin{proof}[Proof of Theorem~\ref{thm:7}]
 For convenience let the metric on the circle~$\cir $ have total length~1.
Choose a periodic coordinate~$t$ so that $\xi =d/dt$ has unit length and is
properly oriented.  Let $L\to \cir $ denote the spin structure, which is a
real line bundle with metric equipped with an isomorphism $L^{\otimes
2}\xrightarrow{\;\cong \;}\underline{\RR}$ of its square with the trivial
bundle of rank one.  Either $L\to \cir $ is the trivial bundle (nonbounding
spin structure) or the M\"obius bundle (bounding spin structure).  The real
skew-adjoint Dirac operator~$D_\phi $ may be identified with the covariant
derivative operator~$\nabla _\xi $ on sections of $L\otimes E\to \cir $.

Suppose first that $L\to \cir $ is the trivial bundle.  If $e\in \Gamma_\phi
$ is a pointwise oriented orthonormal basis of sections of $E\to \cir $, then
$\nabla _\xi (e)=A(e)\cdot e$ for some function $A(e)\:\cir \longrightarrow
\mathfrak{s}\mathfrak{o}\mstrut _{2m}$.  Up to a constant element
of~$SO_{2m}$ we can choose~$e$ so that $A(e)$~is a constant skew-symmetric
matrix~$A\in \mathfrak{s}\mathfrak{o}\mstrut _{2m}$ whose
eigenvalues~$a\sqrt{-1}$ satisfy $-\pi <a\le \pi $.  Then the holonomy of~$E$
around~$\cir $ is $\exp(A)\in SO_{2m}$.  Let $\sH$~be the real Hilbert space
of $L^2$~sections of~$E$ and $W\subset \sH$ the subspace spanned by the
$2m$~sections which comprise the framing~$e$.  The algebraic direct sum
$\bigoplus_{k\in \ZZ}e^{2\pi ikt}W$ is dense in~$\sH$.  Furthermore, the
absolute value of the eigenvalues of~$\nabla _\xi $ on $e^{2\pi ikt}W$ is
bounded below by~$(2|k|-1)\pi$, whence $\nabla _\xi $~is invertible on the
orthogonal complement to~$W$.  It follows directly from the
construction~\cite[\S3]{F3} of the Pfaffian line that $\Pfaff D_\phi
=\Pfaff\nabla _\xi $ is canonically isomorphic to~$\Det W^*$.  There is an
induced isomorphism
  \begin{equation}\label{eq:24}
     \pi _0\bigl(\Pfaff D_\phi\setminus \{0\}\bigr)\xrightarrow{\;\;\cong
     \;\;} \pi _0\bigl(\Det W^*\setminus \{0\} \bigr) 
  \end{equation}
of $\zt$-torsors.  The latter is the $\zt$-torsor of orientations of~$W$.
Now an ordered basis of~$W$ is a sequence of $2m$~sections of $E\to \cir $
which are linearly independent at each point, so after applying Gram-Schmidt
determines an element of~$\Gamma_\phi $.  This induces an isomorphism
  \begin{equation}\label{eq:25}
     \pi _0\bigl(\Det W^*\setminus \{0\} \bigr)\xrightarrow{\;\;\cong \;\;}
     \pi _0(\Gamma_\phi) , 
  \end{equation}
and the isomorphism~\eqref{eq:45} is the composition of \eqref{eq:24}
and~\eqref{eq:25}.

If $L\to \cir $ is the M\"obius bundle, then the preceding argument gives an
isomorphism of the Pfaffian line with the components of the space~$\GL $ of
sections of $SO(E\otimes L)\to \cir $, where $SO(E\otimes L)$~is the oriented
orthonormal frame bundle of ~$E\otimes L$.  Fix a path~$g(t),\;0\le t\le1$,
in~$SO_{2m}$ with $g(0)=1$ and~$g(1)=-1$.  Then if $e\in \GL $~is a section
of~$SO(E)\to \cir $, the product~$e\cdot g$ is a section of~$SO(E\otimes L)$.
There is an induced isomorphism of $\zt$-torsors $\pi
_0(\GL)\xrightarrow{\;\cong \;}\pi _0(\Gamma_\phi) $.
  \end{proof}

  \begin{remark}[]\label{thm:9}
 The isomorphism $\GL\to\Gamma_\phi $ depends on the choice of path~$g$ and
the induced isomorphism $\pi _0(\GL)\to\pi _0(\Gamma_\phi ) $ depends on the
homotopy class of~$g$ rel boundary.  There are two such homotopy classes.
Therefore, the isomorphism of Theorem~\ref{thm:7}, and so the trivialization
of Corollary~\ref{thm:8}, depends on this universal choice.
  \end{remark}

  \begin{remark}[]\label{thm:10}
 Theorem~\ref{thm:7} is an example of a ``categorified index theorem''.  We
expect in general that isomorphisms in theorems of this type depend on
universal choices. 
  \end{remark}

A spin structure on~$X$ leads more directly to a trivialization of the
hamiltonian anomaly~\eqref{eq:19}.  Recalling the discussion
in~\S\ref{subsec:3.2} we solve the quantization problem by the $\zt$-graded
bundle of complex spinors, which is a vector space lift of~\eqref{eq:17}.  In
terms of the bundle of algebras~\eqref{eq:19}, let $\Spin(X)\to X$ denote the
spin structure, a principal $\Spin_{2m}$-bundle.  Left multiplication by
$\Spin_{2m}\subset \Cliff_{2m}$ on~$\Cliff_{2m}$ induces a real vector bundle
over~$X$ which is a bundle of invertible bimodules between~\eqref{eq:22} and
the constant bundle of algebras with fiber~$\Cliff_{2m}$.  (See
~\S\ref{sec:4} for a discussion of the 2-category of algebras; invertible
bimodules are isomorphisms, also known as Morita equivalences.)  Upon
complexification the latter bundle is Morita isomorphic to the trivial bundle
of algebras, since $\Cliff^{\CC}_{2m}$ is Morita trivial.  This Morita
viewpoint on spin structures is emphasized in~\cite{ST}.

  \subsection{The anomaly as an invertible field theory}\label{subsec:3.4}

The modern view in~\S\ref{subsec:2.3} is that the anomaly in supersymmetric
QM is a 2-dimensional invertible extended field theory~$\alpha
_{\textnormal{analytic}}$.  We do not give a direct analytic construction of
the entire field theory from Dirac operators---we have pieces of it in
previous subsections---though that would be an interesting general
undertaking in geometric index theory.  Rather, we use the index theory
carried out in the previous subsections to motivate a direct topological
definition of a field theory $\alpha =\alpha _{\textnormal{topological}}$,
which should be isomorphic to~$\alpha _{\textnormal{analytic}}$.

Recall that the fields~$\sF$ of supersymmetric~QM consist of a metric, spin
structure, map~$\phi $, and fermionic field~$\psi $.  The anomaly in question
occurs after integrating out~$\psi $, so naively we expect it to depend on
the three background fields.  However, as is clear from Theorem~\ref{thm:3}
and the discussion in~\S\ref{subsec:3.2}, it is independent of the metric and
spin structure.\footnote{after some universal choice; see
Remark~\ref{thm:9}.}  Furthermore, up to isomorphism it only depends on the
homotopy class of~$\phi $, since the anomaly is flat: a flat line bundle for
a family of 1-manifolds and a flat gerbe for a family of 0-manifolds.
Therefore, the anomaly has a purely topological description.

Let $\Bord_2(X)$~denote the bordism 2-category of 0-, 1-, and 2-manifolds
equipped with a map to~$X$.  (See~\cite{L} for an exposition of bordism
multicategories and~\cite{Ay} for bordism categories of manifolds with
general geometric structures.)  As the anomaly theory is invertible, it
factors through the geometric realization of~$\Bord_2(X)$, which inverts all
the morphisms.  According to a theorem of
Galatius-Madsen-Tillmann-Weiss~\cite{GMTW}, the result is the 0-space of the
smash product
  \begin{equation}\label{eq:26}
     \Sigma ^2MTO_2\wedge X_+. 
  \end{equation}
Here $MTO_2$~is the Thom spectrum of the virtual vector bundle $-V\to BO_2$,
the negative of the canonical 2-plane bundle over the classifying space
of~$O_2$.  The~`$+$' denotes a disjoint basepoint.  An invertible topological
field theory is a spectrum map out of~\eqref{eq:26}; we take the codomain to
be a shift of the Eilenberg-MacLane spectrum~$H\zt$ for mod~2 cohomology.
(In~\S\ref{subsec:2.3} we discussed a universal choice, the Pontrjagin dual
of the sphere, but for this example the simpler Eilenberg-MacLane spectrum
suffices and captures the theory more precisely.)  That map is the
composition
  \begin{equation}\label{eq:27}
     \alpha \:\Sigma ^2MTO_2\wedge X_+\xrightarrow{\;\;\id\wedge w_2\;\;}
     \Sigma 
     ^2MTO_2\wedge K(\zt,2)_+ \xrightarrow{\;\;\textnormal{Thom}\;\;} \Sigma
     ^2H\zt 
  \end{equation}
To construct the first map we represent the second Stiefel-Whitney class of
the tangent bundle by a map $X\xrightarrow{w_2} K(\zt,2)$ into the
appropriate Eilenberg-MacLane space.  The second spectrum is the Thom
spectrum of $\underline{\RR^2}-V\to BO_2\times K(\zt,2)$, where
$\underline{\RR^2}\to BO_2\times K(\zt,2)$ is the vector bundle with constant
fiber~$\RR^2$.  The Thom isomorphism identifies the second cohomology of the
Thom spectrum with the second cohomology of the (suspension spectrum of the)
base, and the map labeled~`Thom' is the composition of the Thom isomorphism
and projection onto the second factor. Intuitively, if $S$~is a space, then
for~$m=0,1,2$ a map of~$S\times S^m$ into~\eqref{eq:26} is a parametrized
family over~$S$ of closed $m$-manifolds equipped with a map to~$X$:
  \begin{equation}\label{eq:28}
     \begin{gathered} \xymatrix{\sM\ar[d]_\pi\ar[r]^\phi &X \\
     S} \end{gathered}  
  \end{equation}
The value of~$\alpha $, computed as composition with~\eqref{eq:27}, is a map
$S\to K(\zt,2-m)$ whose homotopy class is the transgression~$\pi _*\phi
^*w_2(X)$.

A spin structure on~$X$ can be identified with a null homotopy of the map
$X\to K(\zt,2)$ representing~$w_2(X)$, which induces a null homotopy of the
first map in~\eqref{eq:27} and then, by composition, of~\eqref{eq:27} as
well.  This is a trivialization~\eqref{eq:48} of the anomaly theory~$\alpha
$.

   \section{Central simple algebras and topology}\label{sec:4}
% lastsubsec@  2

Real vector spaces are a model for real $K$-theory in a precise sense, and in
this section we describe models for various truncations of and modules over
real $K$-theory.  We do not give proofs of the statements made here; we hope
to provide them elsewhere.
 
Traditionally~\cite{At2} real $K$-theory is defined on a compact space~$S$ as
the universal abelian group constructed from the monoid of equivalence
classes of real vector bundles on~$S$, with the monoid operation being direct
sum.  Tensor product of vector bundles makes this $K$-theory group into a
ring.  Combining with the suspension construction in topology one obtains
\emph{connective $ko$-theory}, which only has nonzero cohomology in
nonpositive degrees.  Equivalently, the homotopy groups of the spectrum~$ko$
are only nonzero in nonnegative degrees, hence the adjective
`connective'.\footnote{Periodic $KO$-theory is constructed from connective
$ko$-theory by inverting Bott periodicity.  We remark that in the topological
index theory argument of~\S\ref{subsec:3.1} we could have used~$ko$ in place
of~$KO$.}  In somewhat different terms~\cite{Se3}: the 0-space of the
connective spectrum~$ko$ is the classifying space of the symmetric monoidal
topological category of real vector spaces and isomorphisms.  In this section
we introduce other symmetric monoidal topological (multi-)categories and
their classifying connective spectra.

  \subsection{Some $ko$-modules}\label{subsec:4.1}

The nonzero homotopy groups of~$ko$ form the Bott song: 
  \begin{equation}\label{eq:29}
     \pi \mstrut _{0,1,2,\dots  }(ko) =
     \{\ZZ\,,\,\zt\,,\,\zt\,,\,0\,,\,\ZZ\,,\,0\,,\,0\,,\,0\,,\,\ZZ\,,\,\dots      \}.  
  \end{equation}
Just as with spaces, spectra have Postnikov towers and Postnikov
truncations.  For example, the Eilenberg-MacLane spectrum~$\Sigma ^2H\zt$
in~\eqref{eq:27} is the truncation of~$ko$ which just keeps~$\pi _2$.  We
introduce a richer truncation which keeps the first several homotopy groups
  \begin{equation}\label{eq:30}
     R := \pi _{\le4}ko = ko\langle 0\cdots 4 \rangle. 
  \end{equation}
and kills all higher homotopy groups.  This can be done~\cite{B} so that
$R$~is a \emph{ring} spectrum.  Downshifts of~$R$ have negative homotopy
groups, which we truncate by taking connective covers.  For example, we
denote the connective cover of~$\Sigma \inv R$ as~$R\inv $.  Just as we can
consider ordinary cohomology with coefficients in~$\RZ$, there is a
spectrum~$\RRZ$ which represents $R$-cohomology with coefficients in~$\RZ$.
The Postnikov truncation~$R$, its shifts, and their connective covers are all
\emph{module} spectra over the ring spectrum~$ko$.  We also introduce another
module spectrum~$E$, which we define below.  For reference we record here the
nonzero homotopy groups of these spectra:
  \begin{equation}\label{eq:31}
     \begin{tabular}{ c@{\hspace{2em}} c@{\hspace{2em}} c@{\hspace{2em}}
     c@{\hspace{2em}} c@{\hspace{2em}} c@{\hspace{2em}} c@{\hspace{2em}} c}
     \toprule & $R$&  
     $R\inv $ & $\RRZ^{-2}$ & $E$ & $R^{-2}$ & $\RRZ^{-3}$ & $E\inv$ \\
     \midrule \\[-8pt] 
     $\pi_ 4$ & $\ZZ$ &$0 $&$0 $&$0 $&$0 $&$0 $&$0 $\\[-8pt]\\ $\pi_ 3$ & 
     $0$ &$\ZZ $&$0 $&$0 $&$0 $&$0 $&$0 $\\[-8pt]\\ $\pi_ 2$ & $\zt$ &$0
     $&$\RZ$ &$\zt$ &$\ZZ$&$0 $&$0 $\\[-8pt]\\ $\pi_ 1$ & $\zt $ 
     &$\zt $ &$\zt$ & $\zt$ &$0$& $\RZ$ & $\zt$\\[-8pt]\\ $\pi_ 0$ & $\ZZ $
     &$\zt$ &$\zt$ & $\zmod8$ & $\zt$&$\zt$ & $\zt$\\[-8pt] \\ \bottomrule 
     \end{tabular} 
  \end{equation}
This chart also displays the cohomology groups of a point: for any
spectrum~$h$ and~$q\in \ZZ$ we have $h^{-q}(\pt)\cong \pi _qh$.

With the exception of the spectrum~$R$, the 0-space of each spectrum
in~\eqref{eq:31} can be realized as the classifying space of a symmetric
monoidal topological category.  (Perhaps that is true for~$R$ also, but we do
not know any such model.)  The objects in these categories are either lines
or algebras which are $\zt$-graded.  As per custom we use `super' as a
synonym for `$\zt$-graded'.  For example, real metrized superlines---that is,
real inner product spaces of dimension~1 with a $\zt$~grading---are the
objects of a symmetric monoidal category.  Morphisms are degree-preserving
isometries.  The monoidal structure is tensor product, and the symmetry
encodes the Koszul sign rule.  Every object is invertible under tensor
product---invertible vector spaces are lines---and every morphism is also
invertible, thus superlines form a \emph{Picard groupoid}.  It is easy to see
that there are two equivalence classes of objects---the even and odd
line---and that the automorphism group of any object is cyclic of order~2.
The classifying spectrum of this Picard groupoid appears in~\eqref{eq:31}:
it is~$E\inv $.  To prove that statement requires\footnote{given that $E$~has
already been defined!} checking the $k$-invariant between~$\pi _0$ and~$\pi
_1$.  Complex hermitian superlines also form a Picard groupoid.  The group of
equivalence classes of objects is~$\pi _0\cong \zt$ and the automorphism
group of any object is~$\pi _1\cong \RZ$, the latter realized via
exponentiation as the group~$\TT$ of unit norm complex numbers.  Now there
are two choices: we can use the continuous topology or the discrete topology
for the morphisms.  The classifying spectrum with the continuous topology
is~$R^{-2}$ and with the discrete topology it is~$\RRZ^{-3}$.  We term the
latter `flat' since over a space~$S$ the abelian group~ $\RRZ^{-3}(S)$
classifies flat super hermitian complex line bundles.

Fix a field~$k$.  There is a 2-category~$\sC$ whose objects are $k$-algebras,
whose 1-morphisms are bimodules, and whose 2-morphisms are intertwiners.
Invertible 1-morphisms are Morita equivalences, so $\sC$~is sometimes called
the Morita 2-category.  We obtain a Picard 2-groupoid by keeping only
invertible objects and morphisms.  A basic theorem asserts that the
invertible algebras are precisely the central simple algebras, and their
equivalence classes make up the \emph{Brauer group} of~$k$.  The $\zt$-graded
version was proved by Wall~\cite{Wa}; see also~\cite{De}.  Now assume $k=\RR$
or~$k=\CC$.  Just as one can introduce metrics on lines, so too we can
introduce ``metrics'' on invertible superalgebras and invertible
supermodules, and they are used implicitly in the sequel to cut down groups
of 2-automorphisms from~$\CC^\times $ to~$\TT$.  We remark that every central
simple superalgebra over~$\RR$ or~$\CC$ is Morita equivalent to a Clifford
algebra.

The following table summarizes the Picard groupoids of superlines and Picard
2-groupoids of invertible superalgebras and their classifying spectra:
  \begin{equation}\label{eq:32}
     \begin{tabular}{ @{\hspace{3em}}l@{\hspace{3em}} l}
     \toprule  \hspace{-1em}spectrum& Picard (2-)groupoid\\ \midrule \\[-8pt] 
     $R\inv $ & complex central simple superalgebras\\[-8pt]\\ 
     $\RRZ^{-2}$& flat complex central simple superalgebras \\[-8pt]\\ 
     $E$ & real central simple superalgebras\\[-8pt]\\ 
     $R^{-2}$ &complex superlines \\[-8pt]\\ 
     $\RRZ^{-3}$ &flat complex superlines \\[-8pt] \\ 
     $E^{-1}$ & real superlines\\[-8pt] \\ 
     \bottomrule \end{tabular} 
  \end{equation}
The third line can be taken as a definition of~$E$, but the other lines
require proofs, which are fairly routine checks of homotopy groups and
$k$-invariants. 

For the spectra which appear in~\eqref{eq:32} the generalized cohomology
groups of a space~$S$ are equivalence classes of bundles of superlines or
invertible superalgebras.  Thus, for example, $E\inv (S)$~is the abelian
group of real superline bundles up to equivalence.  Bundles of superalgebras,
however, do not suffice to realize all classes in $R\inv (S)$, for example.
We should allow replacement of~$S$ by a locally equivalent
groupoid~\cite[Appendix~A]{FHT} and take fiber bundles of invertible
superalgebras, glued together using fiber bundles of invertible supermodules
and invertible intertwiners.  In this paper we only encounter global bundles
of Clifford algebras, so do not need groupoid replacements.

  \subsection{Some maps between $ko$-modules}\label{subsec:4.2}

Define~$\eta $ as the nonzero element 
  \begin{equation}\label{eq:33}
     \eta \in R^{-1}(\pt)\cong ko^{-1}(\pt) 
  \end{equation}
and $\theta $ as a generator  
  \begin{equation}\label{eq:34}
     \theta \in E^0(\pt) 
  \end{equation}
of the cyclic group~$E^0(\pt)$ of order~8.  They can be represented by
Clifford algebras on a 1-dimensional vector space.  We use the same
symbols~$\eta ,\theta $ for multiplication by these elements.  Let
  \begin{equation}\label{eq:35}
     \beta \mstrut_{\ZZ} \:\RRZ^{-(q+1)}\longrightarrow R^{-q},\qquad q\in \ZZ, 
  \end{equation}
be the connecting homomorphism derived from the fiber sequence $R\to
R_{\RR}\to\RRZ$ of spectra.  Finally, there is a complexification map 
  \begin{equation}\label{eq:36}
     \gamma \:E^{-q}\longrightarrow \RRZ^{-(q+2)},\qquad q=0,1.
  \end{equation}
 
We interpret various multiplication and coboundary maps as geometric
realizations of functors between the Picard groupoids~\eqref{eq:32} and also
the symmetric monoidal groupoid of real vector spaces, whose classifying
spectrum is~$ko$.

  \begin{proposition}[]\label{thm:11}
 \ 

  \begin{enumerate}
 \item The assignment of the real Clifford algebra~$\Cliff(V)$ to a real
vector space~$V$ induces the spectrum map $\theta \:ko\to E$,
multiplication by~\eqref{eq:34}.

 \item The assignment of the complex Clifford algebra~$\Cliff^{\CC}(V)$ to a
real vector space~$V$ induces the spectrum map $\eta \:ko\to R $,
multiplication by~\eqref{eq:33}.

 \item The assignment of the complexification~$A_{\CC}$ to a real central
superalgebra~$A$ induces the spectrum map $\gamma \:E\to \RRZ^{-2}$
in~\eqref{eq:36}.

 \item The assignment of the complexification~$L_{\CC}$ to a real
superline~$L$ induces the spectrum map $\gamma \:E^{-1}\to \RRZ^{-3}$
in~\eqref{eq:36}.

 \item The spectrum map $\beta \mstrut_{\ZZ} \:\RRZ^{-2}\to R\inv $ forgets
the flat structure.

 \item The spectrum map $\beta \mstrut_{\ZZ} \:\RRZ^{-3}\to R^{-2}$ forgets
the flat structure.

 \item The Postnikov truncation $ko\inv \to \pi _{\le1}ko\inv
\xrightarrow{\;\cong \;} E\inv $ is multiplication by~$\theta $
in~\eqref{eq:34}.

% \item The identity $\beta \mstrut _{\ZZ}\gamma \theta =\eta$ holds.

  \end{enumerate}
  \end{proposition}

\noindent
 For example, if $S$~is a space and $L\to S$ a real line bundle over~$S$ with
equivalence class~$[L]\in E\inv (S)$, then (iv)~asserts that $\gamma [L]\in
\RRZ^{-3}(S)$ is the equivalence class of the complexification $L_{\CC}\to
S$, which carries a natural flat structure, and (vi)~asserts that $\beta
\mstrut_{\ZZ} \gamma [L]\in R^{-2}(S)$ is the equivalence class of the
complex line bundle $L_{\CC}\to S$ if we disregard the flat structure.  We
remark that statement~(iii) can be used as the definition of the map~$\gamma
$.

  \begin{remark}[]\label{thm:12}
 There is one more theorem of this kind which is relevant here.  According
to~\cite{ABS} elements of~$ko\inv (\pt)$ are represented by supermodules over
the Clifford algebra~$\Cliff(\RR)$ with a single generator of square~$-1$.
To such a module~$W^0\oplus W^1$ we assign the real superline~$\Det (W^0)^*$,
which is even or odd according to $\dim W^0\pmod2$.  This induces a map
$ko^{-1}\to E^{-1}$.  The theorem is that this map is the one in~(vii).  This
construction relates~(vii) to the Pfaffian superline bundle of a family of
Dirac operators on 1-manifolds, if we use the Clifford linear Dirac
operator~\cite{LM}.  
  \end{remark}

   \section{Supersymmetric QM with a general target}\label{sec:5}
% lastsubsec@000

We revisit anomalies in supersymmetric~QM, only now we drop
Assumption~\ref{thm:2}.  Thus the target~$X$ is an arbitrary Riemannian
manifold.  Supersymmetric~QM is still defined; the fields and
lagrangian~\eqref{eq:9} are unchanged.  It still makes sense to integrate out
the fermionic field~$\psi $ to obtain a relative theory.  In this section we
identify the anomaly theory~$\alpha $, which is an invertible extended
2-dimensional topological field theory. 
 
It is easiest to begin with the hamiltonian anomaly, which is the value
of~$\alpha $ on a point.  The discussion in~\S\ref{subsec:3.2} carries over:
the space of classical solutions is still the supersymplectic
manifold~\eqref{eq:15}.  However, if $X$~is odd dimensional there is no
polarization and if $X$~is not oriented there is no oriented polarization.
Instead we consider quantization from the operator algebra viewpoint.
Namely, at each point of~$X$ the operator algebra in the fermionic system
with field~$\psi $ is the complex Clifford algebra~$\Cliff^{\CC}(T_xX)$.  In
the family of fermionic system parametrized by constant maps~$\phi $
into~$X$, the family of operator algebras is the bundle~\eqref{eq:19} of
complex Clifford algebras.  A quantization is a complex vector bundle ~$E\to
X$ and an isomorphism of $\Cliff^{\CC}(TX)\to X$ with the bundle of
endomorphisms $\End(E)\to X$ Furthermore, the Riemannian metric on~$X$
induces a metric structure\footnote{We alluded to this type of metric
structure before~\eqref{eq:32}, but did not define it.  It is something we
expect in a \emph{unitary} quantum field theory.} which in this case is flat
and in fact is induced from the bundle~\eqref{eq:22} of real Clifford
algebras.  Applying~\eqref{eq:32} and Proposition~\ref{thm:11} we conclude
that the equivalence class of the \emph{flat} bundle~\eqref{eq:19} of complex
central simple superalgebras is
  \begin{equation}\label{eq:37}
     [\Cliff^{\CC}(TX)] = \gamma \theta [TX]\quad \in \RRZ^{-2}(X). 
  \end{equation}
This is hamiltonian anomaly: the obstruction to finding the vector bundle
$E\to X$.
 
We can also analyze the lagrangian anomaly.  The real Pfaffian line
bundle~\eqref{eq:10} is defined as in~\S\ref{subsec:3.1}, but now it is
$\zt$-graded by the mod~2 index.  (Under Assumption~\ref{thm:2} the Pfaffian
line bundle is even, so we did not encounter the $\zt$-grading previously.)
The following result is expressed in terms of transgression
using~\eqref{eq:38}.  Recall that we are studying a family of real
skew-adjoint Dirac operators on a spin 1-manifold~$M$.

  \begin{theorem}[]\label{thm:13}
 The topological equivalence class in $E^{-1}\bigl(\sF'(M)\bigr)$ of the
lagrangian anomaly\newline $\Pfaff D\to\sF'(M)$ is $\theta \,(\pi
_1)_*e^*[TX]$.
  \end{theorem}

\noindent
 Here $[TX]\in ko^0(X)$ is the $ko$-theory class of the tangent bundle
of~$X$, and the pushforward~$(\pi _1)_*$ in $ko$-theory uses the spin
structure on~$M$. 

  \begin{proof}
 The Atiyah-Singer topological index theorem~\cite{AS2} identifies $(\pi
_1)_*e^*[TX]\in ko^{-1}\bigl(\sF'(M) \bigr)$ as the index of the family of
Dirac operators.  The Pfaffian line bundle is computed by the lowest 2-stage
Postnikov truncation of~$ko$, and Proposition~\ref{thm:11}(vii) implies that
it is computed as multiplication by~$\theta $.  
  \end{proof}

\noindent
 See Remark~\ref{thm:12} for a more direct relationship between the Pfaffian
line bundle and the $ko$-index. 

  \begin{remark}[]\label{thm:18}
 By Proposition~\ref{thm:11}(iv) the equivalence class of the flat complex
superline bundle obtained from the real Pfaffian superline bundle is 
  \begin{equation}\label{eq:46}
     \gamma \theta (\pi _1)_*e^*[TX]\quad \in \RRZ^{-3}\bigl(\sF'(M) \bigr). 
  \end{equation}
\end{remark}

Motivated by~\eqref{eq:37} and~\eqref{eq:46} we posit a direct definition of
the anomaly field theory~$\alpha $, as in~\S\ref{subsec:3.4}.  In this
general case we have already seen in Theorem~\ref{thm:13} that the spin
structure on ``spacetime''~$M$ enters, so we expect a theory on the bordism
2-category~$\Sbord_2(X)$ of 0-, 1-, and 2-dimensional spin manifolds equipped
with a map to~$X$.  The Madsen-Tillmann spectrum $\Sigma ^2MT\Spin_2\wedge
X_+$ is its geometric realization, which now replaces~\eqref{eq:26}, and we
let~$\alpha $ take values in the spectrum~$\Sigma ^{-2}\RRZ$ whose 0-space
classifies flat complex central simple superalgebras.  Analogous
to~\eqref{eq:27}, we define~ $\alpha $ as the composition
  \begin{equation}\label{eq:39}
     \alpha \:\Sigma ^2MT\Spin_2\wedge X_+\xrightarrow{\;\;\id\wedge [TX]\;\;}
     \Sigma ^2MT\Spin_2\wedge (ko\mstrut _0)_+ \xrightarrow{\;\;\theta \,\circ\,
     \textnormal{Thom}\;\;} E\xrightarrow{\;\;\gamma \;\;}     \Sigma ^{-2}\RRZ
  \end{equation}
We have chosen a map $X\to ko\mstrut _{0}$ into the 0-space of the $K$-theory
spectrum which represents~$[TX]\in ko^0(X)$.  The second map in~\eqref{eq:39}
is the composition of the Thom isomorphism in $ko$-theory for spin
bundles~\cite{ABS}, a projection map, and multiplication by~$\theta $.  Since
$\theta $~commutes with transgression---it is pulled back from a point---we
can rewrite~\eqref{eq:39} by first acting by~$\gamma \theta $ and then
applying the Thom isomorphism for the theory~$\RRZ$:
  \begin{equation}\label{eq:40}
     \alpha \:\Sigma ^2MT\Spin_2\wedge X_+\xrightarrow{\;\;\id\wedge
     \gamma \theta [TX]\;\;} \Sigma ^2MT\Spin_2\wedge ((\RRZ)\mstrut _{-2})_+
     \xrightarrow{\;\; \textnormal{Thom}\;\;} \Sigma ^{-2}\RRZ 
  \end{equation}

Suppose $W$~is a closed 2-manifold with spin structure~$\sigma $ and a smooth
map $\phi \:W\to X$.  We compute $\alpha (W,\sigma ,\phi )\in \RZ$.  The
map~$\gamma $ in~\eqref{eq:39} simply includes $E^{-2}(\pt)\cong
\zt\hookrightarrow \RZ\cong \RRZ^{-4}(\pt)$ when evaluated on a 2-manifold.
Furthermore, by the looping of Proposition~\ref{thm:11}(vii) the map~$\theta
$ does nothing in this case.  Hence we identify
  \begin{equation}\label{eq:43}
     \alpha (W,\sigma ,\phi ) = \pi ^W_* \phi ^*[TX]\quad \in
     ko^{-2}(\pt)\cong \zt, 
  \end{equation}
where $\pi ^W\:W\to\pt$ and we pushforward in $ko$-theory using the spin
structure on~$W$.  The Atiyah-Singer index theorem identifies this as the
mod~2 index of the Dirac operator on~$W$ coupled to~$\phi ^*TX$.  Denote
$\Arf_W(\sigma )= \pi ^W_*(1)\in \zt$, where $1\in ko^0(W)$ is the unit.
$\Arf_W$~is a quadratic function on spin structures~\cite{At3}.  Set
$w_q=w_q(X)$, $q=1,2$.

  \begin{proposition}[]\label{thm:16}
 $\alpha (W,\sigma ,\phi ) = (\dim X)\Arf_W(\sigma ) + \Arf_W(\sigma +\phi
^*w_1) + \Arf_W(\sigma ) + \langle \phi ^*w_2,[W]  \rangle$.
  \end{proposition}

  \begin{proof}
 Write $[TX] = \dim X + ([TX] - \dim X)$ to pick off the first term and
reduce to evaluating the $ko$-pushforward on a virtual bundle of rank zero.
That bundle can be written as $[L_{w}]-1$ plus a class~$z\in ko^0(W)$ of rank
zero with vanishing first Stiefel-Whitney class, where $L_w\to W$~is the real
line bundle with Stiefel-Whitney class $w=\phi ^*(w_1)$.  The pushforward
of~$[L_w]$ in spin structure~$\sigma $ equals the pushforward of~$1$ in spin
structure~$\sigma +w$.  This explains the middle two terms in the formula.
Finally, the class~$z$ can be represented by a map from~$W$ into the
2-skeleton of~$ko\mstrut _0$ which is trivial on the first Postnikov section,
so a map into~$K(\zt,2)$.  That map can be taken to be trivial off of a ball
in~$W$, and since the $ko$-pushforward of that class is easily seen to be
independent of spin structure, by the bordism invariance of the pushforward
we can replace~$W$ by a 2-sphere.  Now the pushforward is the suspension
isomorphism, and the last term in the formula results.
  \end{proof}

  \begin{remark}[]\label{thm:17}
 The cobordism hypothesis~\cite{L} asserts that the extended topological
field theory~$\alpha $ is determined by its value~\eqref{eq:37} on a point.
But the cobordism hypothesis is overkill for an invertible topological theory
as we can define it directly by specifying the map~\eqref{eq:40} (or
equivalently~\eqref{eq:39}) of spectra.
  \end{remark}

Finally, we discuss trivializations of the anomaly theory~$\alpha $.  As
described at the end of~\S\ref{subsec:3.3} a spin structure on~$X$
trivializes the hamiltonian anomaly as long as $X$~is even-dimensional.  That
still applies to~\eqref{eq:37}.  If $X$~is odd-dimensional and spin, then the
spin structure induces a Morita isomorphism of the flat bundle~\eqref{eq:19}
of complex Clifford algebras with the constant bundle whose fiber is the
complex Clifford algebra~$\Co=\Cliff^{\CC}(\RR)$ on a single generator.  We
can interpret this as saying that the bundle of Hilbert spaces over~$X$
obtained by quantizing~$\psi $ is naturally a bundle of~$\Co$-modules.
Furthermore, quantizing~$\phi $ we see that the Hilbert space of
supersymmetric~QM is also naturally a $\Co$-module.  Should we say that the
theory is anomalous, or allow that the Hilbert space of a quantum theory can
be a $\Co$-module?  I opt for the latter.

  \begin{remark}[]\label{thm:15}
 This fixes a well-known problem about fermions on an odd-dimensional
manifold.  For example, path integral arguments~\cite[p.~682]{Detal} suggest
that the dimension of the Hilbert space is an integer multiple of~$\sqrt{2}$
if $\dim X$~is odd.  We see here that the Hilbert space is naturally a
$\Co$-module, which resolves this puzzle with the path integral.
  \end{remark}

We can see directly from~\eqref{eq:40} the effect of a spin structure on~$X$
on the entire anomaly theory~$\alpha $.  A spin structure from this point of
view is a homotopy of the map $X\to E_0$ representing~$\theta [TX]$ to a
constant map into some component of the 0-space of~$E_0$.  Running this
homotopy through the composition~\eqref{eq:40} we obtain a homotopy
from~$\alpha $ to either (i)~the trivial theory if $\dim X$~is even, or (ii)~
a particularly simple 2-dimensional invertible extended topological
theory~$\alpha '$ if $\dim X$~is odd.  The theory~$\alpha '$ assigns the
Clifford algebra~$\Co$ to a point; the even or odd line to a spin circle,
depending on whether the spin structure bounds or not; and the Arf
invariant~$\Arf_W(\sigma )$ to a closed spin 2-manifold~$W$ with spin
structure~$\sigma $.  (See~\cite{G} for more on~$\alpha '$ together with an
interesting geometric application.)

\newcommand{\etalchar}[1]{$^{#1}$}
\providecommand{\bysame}{\leavevmode\hbox to3em{\hrulefill}\thinspace}
\providecommand{\MR}{\relax\ifhmode\unskip\space\fi MR }
% \MRhref is called by the amsart/book/proc definition of \MR.
\providecommand{\MRhref}[2]{%
  \href{http://www.ams.org/mathscinet-getitem?mr=#1}{#2}
}
\providecommand{\href}[2]{#2}

  \end{document}